\documentclass[final,10pt,3p,times,twocolumn]{elsarticle}

\usepackage{microtype}

\usepackage{amsmath}

\usepackage{amsthm}
\usepackage{cleveref}

\theoremstyle{plain}
\newtheorem{thm}{Theorem}[section]
\newtheorem{lem}[thm]{Lemma}

\newtheorem{cor}[thm]{Corollary}

\crefname{section}{Section}{Sections}
\crefname{figure}{Figure}{Figures}
\crefname{lem}{Lemma}{Lemmas}
\crefname{thm}{Theorem}{Theorems}
\crefname{cor}{Corollary}{Corollaries}
\crefname{obs}{Observation}{Observations}
\crefname{equation}{}{}

\newcommand{\F}{\mathcal{F}}
\renewcommand{\L}{\mathcal{L}}
\newcommand{\Q}{\mathcal{Q}}
\newcommand{\T}{\mathcal{T}}

\renewcommand\ncong{\not{ }\cong}

\usepackage{tikz}
\usetikzlibrary{calc}

\newif\ifnotes
\notestrue
\ifnotes
  \usepackage{todonotes,tikz}
  \newcommand\nz[2][]{\todo[backgroundcolor=yellow,linecolor=black,bordercolor=black,caption={},#1]{\raggedright NZ: #2}}
  \newcommand\mj[2][]{\todo[backgroundcolor=cyan,linecolor=black,bordercolor=black,caption={},#1]{\raggedright MJ: #2}}
  \newcommand\nh[2][]{\todo[backgroundcolor=green,linecolor=black,bordercolor=black,caption={},#1]{\raggedright NH: #2}}
  \newcommand\leo[2][]{\todo[backgroundcolor=green,linecolor=black,bordercolor=black,caption={},#1]{\raggedright LEO: #2}}
  
  \AddToHook{shipout/background}{\put(\paperwidth-3in,-\paperheight){\begin{tikzpicture}
    \path [fill=black!10] (0,0) rectangle (3in,\paperheight);
  \end{tikzpicture}}%
  \put(0in,-\paperheight){\begin{tikzpicture}
    \path [fill=black!10] (0,0) rectangle (3in,\paperheight);
  \end{tikzpicture}}}
\else
  \newcommand\nz[2][]{}
  \newcommand\mj[2][]{}
  \newcommand\nh[2][]{}
  \newcommand\leo[2][]{}
  
\fi

\journal{Information Processing Letters}

\begin{document}

\begin{frontmatter}
\title{A Simple 4-Approximation Algorithm for Maximum Agreement Forests on
Multiple Unrooted Binary Trees\tnoteref{fund1}}
\author[dal]{Jordan Dempsey}
\author[delft]{Leo van Iersel}
\author[delft]{Mark Jones}
\author[dal]{Norbert Zeh}

\affiliation[dal]{organization={Faculty of Computer Science, Dalhousie University},
            addressline={6050 University Avenue}, 
            city={Halifax},
            postcode={B3H 1W5}, 
            state={Nova Scotia},
            country={Canada}}

\affiliation[delft]{organization={Delft Institute of Applied Mathematics, Delft University of Technology},
            addressline={Mekelweg 4}, 
            city={Delft},
            postcode={2628 CD}, 
            country={the Netherlands}}

\tnotetext[fund1]{This paper received
funding from the Netherlands Organisation for Scientific Research (NWO) under
projects OCENW.M.21.306 and OCENW.KLEIN.125, and from the Natural Sciences and Engineering Research Council of Canada under grant RGPIN/05435-2018.}

\date\today

\begin{abstract}\noindent We present a simple 4-approximation algorithm for
  computing a maximum agreement forest of multiple unrooted binary trees.  This
  algorithm applies LP rounding to an extension of a recent ILP formulation of
  the maximum agreement forest problem on two trees by Van Wersch
  al.~\cite{van}.  We achieve the same approximation ratio as the algorithm of
  Chen et al.~\cite{chen} but our algorithm is extremely simple.  We also prove
  that no algorithm based on the ILP formulation by Van Wersch et al.\ can
  achieve an approximation ratio of $4 - \varepsilon$, for any $\varepsilon >
  0$, even on two trees.  To this end, we prove that the integrality gap of the
  ILP approaches 4 as the size of the two input trees grows.
\end{abstract}


\end{frontmatter}

\section{Introduction} 

Phylogenetic trees (and networks) model the evolution of a set of taxa.
Different methods for constructing such trees from, say, DNA data may produce
different results.  Even the same method may produce different results when
trees on the same set of taxa are constructed from different genes shared by
these taxa.  As a result, it has become important to measure the
dissimilarity between phylogenetic trees, both to quantify confidence in the
trees constructed using different methods and to discover non-tree-like events
in the evolution of a set of taxa that explain the differences between trees
constructed from different genes.

One distance measure used to quantify the dissimilarity between two unrooted
phylogenetic trees is the \emph{tree bisection and reconnection} (TBR) distance
\cite{allensteel}.  A~TBR operation on a tree $T$ removes an edge $\{u,
v\}$ of~$T$, thereby splitting $T$ into two subtrees $T_u$ and $T_v$; supresses
$u$ an $v$, as they now have degree $2$; subdivides an edge in $T_u$ and an edge
in $T_v$; and then reconnects $T_u$ and $T_v$ by adding an edge between the two
vertices introduced by subdividing these two edges.  The TBR distance between
two trees is the number of such TBR operations necessary to turn one of the two
trees into the other. This distance is known to be one less than the
size of a \emph{maximum agreement forest (MAF)} of the two trees
\cite{allensteel}.

An \emph{agreement forest} (AF) of a set of trees $\T$ is a forest that can be
obtained from each tree in $\T$ by deleting edges and suppressing degree-$2$
vertices.  A \emph{maximum agreement forest} (MAF) is an agreement forest with
the minimum number of components (which corresponds to preserving the maximum
number of edges in each tree in $\T$).  While the TBR distance is
difficult to extend to more than two trees, the definition of a MAF does
generalize naturally to more than two trees and is meaningful as a measure of
(dis)similarity of the given set of trees, as it captures the parts of the
evolutionary history of a set of taxa on which all input trees agree.

Computing a MAF is NP-hard even for two trees \cite{allensteel, hein}.  This
motivates the study of parameterized and approximation algorithms for computing
MAFs.  The best known kernel for this problem, due to Kelk et
al.~\cite{kelk2022}, has size $9k - 8$, where $k$ is the size of the MAF.
Hallett and McCartin provided a branching algorithm for the same problem with
running time $O(4^k \cdot k^5 + n^{O(1)})$ \cite{hallett}. Chen et al.\ further
improved this bound to $O(3^k \cdot n)$ \cite{chen2013}.  Van Wersch et al.\
provided a new ILP formulation of the MAF problem, as well as improved
kernelization results that were incorporated into the Tubro software for
computing TBR distance~\cite{van}. 

Whidden and Zeh presented a linear-time $3$-approximation algorithm for
the unrooted MAF problem on two trees \cite{whidden}, which remains the best
known algorithm for this problem.  Chen et al.\ \cite{chen} presented a
$4$-approximation algorithm for multiple binary trees.  Their algorithm is
purely combinatorial but is rather complicated --- significantly more
complicated than the $3$-approximation algorithm for two trees.

In this paper, we show that the ILP formulation of the MAF problem by Van Wersch
et al.\ \cite{van} can be combined with an extremely simple LP rounding
approach to match the approximation ratio achieved by Chen et al.\ on multiple
binary trees.  We also show that no algorithm based on this ILP formulation can
achieve an approximation ratio of $4 - \varepsilon$ for any constant
$\varepsilon > 0$, even when restricted to two trees.  We do this by proving
that the integrality gap of this ILP formulation approaches 4 as the size of the
two input trees grows.

The remainder of this paper is organized as follows: \cref{sec:preliminaries}
provides formal definitions of the concepts used in this paper.  \cref{sec:ilp}
provides an ILP formulation of the problem of computing a MAF of a set of trees.
\cref{sec:rounding} provides a 4-approximation algorithm for computing a MAF of
a set of trees based on this ILP formulation.  \cref{sec:tight-inputs} proves
that the integrality gap of this ILP is $4 - o(1)$, thus precluding a better
approximation ratio than $4$ for any algorithm based on this~ILP.
\cref{sec:conclusions} offers concluding remarks and explains the significance
of our results for related problems.

\section{Preliminaries}

\label{sec:preliminaries}

A \textit{(binary) phylogenetic tree} $T$ is a tree whose internal
vertices have degree $3$ and are unlabelled, and whose leaves are
labelled bijectively with the elements of some set~$X$.  We call the elements of
$X$ \emph{taxa} and do not distinguish between a leaf and its label.  All trees
in this paper are unrooted, that is, they are connected undirected graphs
without cycles.

Two phylogenetic trees $T_1$ an $T_2$ over the same leaf set $X$ are
\emph{isomorphic} (written $T_1 \cong T_2$) if there exists a graph isomorphism
$\phi : T_1 \cong T_2$ such that $\phi(x) = x$, for all $x \in X$ (i.e., $\phi$
respects the leaf labels).

For any subset $Y \subseteq X$ we denote by $T[Y]$ the minimal subtree of $T$
that connects all leaves in $Y$.  We use $T|_Y$ to refer to the tree obtained by
suppressing all degree-$2$ vertices in $T[Y]$.  \emph{Suppressing} a degree-$2$
vertex $v$ with neighbours $u$ and $w$ is the operation of removing $v$ and its
incident edges and reconnecting $u$ and $w$ with an edge $\{u, w\}$.

For a set $\T = \{T_1, \ldots, T_t\}$ of phylogenetic trees over the same leaf
set $X$, an \emph{agreement forest} (AF) of $\T$ is a partition $\F = \{Y_1,
\ldots, Y_k\}$ of $X$ such that\footnote{This definition follows Linz and Semple
\cite{linz-cluster}.  Replacing $\F = \{Y_1, \ldots, Y_2\}$ with $\F =
\{T_1|_{Y_1}, \ldots, T_1|_{Y_2}\}$ produces the equivalent definition used by
Allen and Steel \cite{allensteel}, where $\F$ is an actual forest.}
\begin{enumerate}
    \item For all $1 \le i \le j \le t$ and $1 \le h \le k$, $T_i|_{Y_h} \cong
    T_j|_{Y_h}$ and 
    \item For all $1 \le i \le t$ and $1 \le h < h' \le k$, $T_i[Y_h]$ and
    $T_i[Y_{h'}]$ are disjoint subtrees of $T_i$.
\end{enumerate}
This captures the intuitive definition of an AF given in the introduction:
Condition (2) expresses that the trees $T_i[Y_1], \ldots, T_i[Y_t]$ are
separated by edges in~$T_i$, so the set of these trees can be obtained from
$T_i$ by cutting some set of edges in $T_i$.  Condition (1) expresses that we
must obtain the same collection of trees from every tree in $\T$ after
suppressing degree-$2$ vertices.  We call $Y_1, \ldots, Y_k$ the \emph{components}
of $\F$.  We say that two distinct components $Y_h, Y_{h'} \in \F$ \emph{overlap
in $T_i$} if they violate condition (2) for $T_i$ (in this case, $\F$ is not an
AF of $\T$).  We say that $\F$ is a \emph{maximum agreement forest} (MAF) of
$\T$ if there is no AF $\F'$ of $\T$ of size $|\F'| < |\F|$.

A \emph{quartet} is a subset of $X$ of size $4$.  For a quartet $Q = \{a, b, c,
d\}$, let $ab|cd$ be the tree with leaf set $Q$ in which $a$ and $b$ share a
common neighbour $u$, $c$~and $d$ share a common neighbour $v$, and $u$ and $v$
are connected by an edge.  If $T_1|_Q \cong ab|cd$, then we define $\L(Q)$ to
be the set of edges in $T_1[\{a, b\}] \cup T_1[\{c, d\}]$.  A quartet $Q
\subseteq X$ is an \emph{incompatible} quartet of two trees $T_i, T_j \in \T$ if
$T_i|_Q \ncong T_j|_Q$.

\begin{lem}
  \label{lem:quartet-isomorphic}Two phylogenetic trees $T_1$ and $T_2$ on the same leaf set $X$ are isomorphic if and only if they have no incompatible quartets.
\end{lem}
\cref{lem:quartet-isomorphic} follows as a result of the work presented by Colonius and 
Schulze who first gave this characterisation of quartet trees in 1977 \cite{colonius77, colonius81}.
\section{An ILP Formulation of the MAF problem}

\label{sec:ilp}

Let $\T = \{T_1, \ldots, T_t\}$ be a set of $t$ phylogenetic trees over the same
label set $X$.  For $2 \le i \le t$, let $\Q_i$ be the set of incompatible
quartets of $T_1$ and $T_i$.  Let $\Q = \bigcup_{i=2}^t \Q_i$.  We prove that an
optimal solution to the following ILP defines a MAF of $\T$:

\begin{equation}
  \begin{gathered}
    \textrm{Minimize}\ \sum_{e \in E(T_1)} x_e\\
    \begin{aligned}
      \textrm{s.t.}\ \sum_{e \in \L(Q)} x_e &\ge 1 && \forall Q \in \Q\\
      x_e &\in \{0, 1\} && \forall e \in E(T_1).
    \end{aligned}
  \end{gathered}
  \label{eq:ilp}
\end{equation}

Any solution $\hat x$ of \cref{eq:ilp} defines a set $E_{\hat x} = \{e \in
E(T_1) \mid \hat x_e = 1\}$.  The mapping $\hat x \mapsto E_{\hat x}$ is
easily seen to be a bijection between the set of solutions of \cref{eq:ilp} and
the set of subsets of $E(T_1)$.  Thus, we mostly do not distinguish
between solutions of \cref{eq:ilp} and subsets of $E(T_1)$.

Any subset of edges $E \subseteq E(T_1)$ defines a partition $\F_E = \{Y_1,
\ldots, Y_k\}$ of $X$ where two leaves $a, b \in X$ belong to the same component
$Y_h$ if and only if $T_1[\{a, b\}] \cap E = \emptyset$.

\begin{thm}
  \label{thm:ilp}
  A subset $E \subseteq E(T_1)$ is a feasible solution of \cref{eq:ilp} if and
  only if $\F_E$ is an agreement forest of\/~$\T$.
\end{thm}

\begin{proof}
  The proof follows the proof for two trees \cite{van}.

  First assume $\F_E$ is not an AF of $\T$.  Then there exist either a
  tree $T_i \in \T$ and two components $Y_h, Y_{h'} \in \F_E$ that overlap in
  $T_i$, or two trees $T_i, T_j \in \T$ and a component $Y_h \in \F_E$ such that
  $T_i|_{Y_h} \ncong T_j|_{Y_h}$.
  
  In the latter case, we can assume w.l.o.g.\ that $T_1|_{Y_h} \ncong
  T_i|_{Y_h}$ because $T_i|_{Y_h} \ncong T_j|_{Y_h}$ implies that we cannot have
  both $T_1|_{Y_h} \cong T_i|_{Y_h}$ and $T_1|_{Y_h} \cong T_j|_{Y_h}$.  By
  \cref{lem:quartet-isomorphic}, this implies that there exists a quartet $Q
  \subseteq Y_h$ with $T_1|_Q \ncong T_i|_Q$. Since $Q \subseteq Y_h$,
  we have $\L(Q) \cap E = \emptyset$, so $E$ is not a feasible solution of
  \cref{eq:ilp}.
  
  If two components $Y_h, Y_{h'} \in \F_E$ overlap in $T_i$, then there exist
  two leaves $a, b \in Y_h$ and two leaves $c, d \in Y_{h'}$ such that the two
  paths $T_i[\{a, b\}]$ and $T_i[\{c, d\}]$ share an edge.  Thus, for $Q = \{a,
  b, c, d\}$, $T_i|_Q \cong ac|bd$ or $T_i|_Q \cong ad|bc$.  On the other hand,
  $T_1[\{a, b\}] \cap E = \emptyset$ and $T_1[\{c, d\}] \cap E =
  \emptyset$ because $a, b \in Y_h$ and $c, d \in Y_{h'}$, and $T_1[\{x, y\}]
  \cap E \ne \emptyset$, for all $x \in \{a, b\}$ and $y \in \{c, d\}$ because
  $x \in Y_h$, $y \in Y_{h'}$, but $Y_h \ne Y_{h'}$.  This implies that $T_1|_Q
  \cong ab|cd \ncong T_i|_Q$ and that $\L(Q) \cap E = \emptyset$, so once again,
  $E$ is not a feasible solution of \cref{eq:ilp}.
  
  Now assume $E$ is not a feasible solution of~\cref{eq:ilp}.  Then there exists
  a quartet $Q = \{a, b, c, d\} \in \Q$ such that $T_1|_Q \cong ab|cd$
  and $E \cap \L(Q) = \emptyset$.  Assume $Q \in \Q_i$.  Since $E \cap \L(Q) =
  \emptyset$, $a$~and $b$ belong to the same component $Y_h$ of $\F_E$, and $c$
  and $d$ belong to the same component $Y_{h'}$ of $\F_E$.  Since $T_1|_Q \cong
  ab|cd$ and $Q$ is an incompatible quartet of $T_1$ and $T_i$, the paths
  $T_i[\{a, b\}]$ and $T_i[\{c, d\}]$ share an edge.  Thus, if $Y_h \ne Y_{h'}$,
  these two components overlap in $T_i$, and $\F_E$ is not an AF of $\T$.  If
  $Y_h = Y_{h'}$, then $T_i|_Q \cong ac|bd$ or $T_i|_Q \cong ad|bc$.  In either case,
  $T_1|_Q \ncong T_i|_Q$.  Thus, by \cref{lem:quartet-isomorphic},
  $T_1|_{Y_h} \ncong T_i|_{Y_h}$ and, once again, $\F_E$ is not an AF of $\T$.
\end{proof}

It is easily verified that every AF $\F$ of $\T$ satisfies $\F = \F_E$, for some
subset of edges $E \subseteq E(T_1)$, that any such set has size $|E| \ge
|\F| - 1$, and that there exists such a set of size $|E| = |\F| - 1$. Thus,
since $|E_{\hat x}| = \sum_{e \in E(T_1)} \hat x_e$, \cref{thm:ilp} immediately
implies the following corollary.

\begin{cor}
  \label{cor:main}
  A subset $E \subseteq E(T_1)$ is an optimal solution of \cref{eq:ilp} if and
  only if $\F_E$ is a maximum agreement forest of\/ $\T$ and $|E| = |\F_E|
  - 1$.
\end{cor}

\section{A 4-Approximation Based on LP Rounding}

\label{sec:rounding}

Now consider the LP relaxation of \cref{eq:ilp}, where the constraint $x_e \in
\{0, 1\}$ is replaced with the constraint $x_e \ge 0$.  Let $\tilde x$ be an
optimal fractional solution of this LP relaxation.  By \cref{thm:ilp}, any
integral feasible solution $\hat x$ of \cref{eq:ilp} corresponds to an AF
$\F_{E_{\hat x}}$ of $\T$ of size $|\F_{E_{\hat x}}| \le |E_{\hat x}| + 1
= \sum_{e \in E(T_1)} \hat x_e + 1$.  By \cref{cor:main}, any optimal integral
solution $x^*$ of \cref{eq:ilp} corresponds to a MAF $\F_{E_{x^*}}$ of size
$|\F_{E_{x^*}}| = |E_{x^*}| + 1 = \sum_{e \in E(T_1)} x^*_e + 1 \ge \sum_{e \in
E(T_1)} \tilde x_e + 1$. Thus, if $\sum_{e \in E(T_1)} \hat x_e \le
4\sum_{e \in E(T_1)} \tilde x_e$, then $|\F_{E_{\hat x}}| \le
4|\F_{E_{x^*}}|$, that is, $\F_{E_{\hat x}}$ is a 4-approximation of a MAF of
$\T$.  Such a solution can be obtained as follows.  It is more intuitive to
describe the construction in terms of the edge set $E = E_{\hat x}$.

We compute an optimal fractional solution $\tilde x$ of \cref{eq:ilp}, choose an
arbitrary leaf $r$ of $T_1$ as its root, and initially set $E = \emptyset$.  For
every edge $e$, we name its endpoints $u_e$ and $v_e$ such that $u_e$ is on the
path from $r$ to $v_e$ (i.e., $u_e$ is the parent of $v_e$).  Let $D(e)$ be the
set of descendant edges of $e$ that belong to the same connected component of
$T_1 - E$ as $v_e$.  Formally, $f \in D(e)$ if $v_e$ belongs to the path from
$r$ to $v_f$ and the path from $v_e$ to $v_f$ contains no edge in $E$.  Finally,
let $w(e) = \sum_{f \in D(e)} \tilde x_f$.  Now we choose an edge $e$ such that
\begin{equation}
  w(e) \ge 1/4 \text{ but } w(f) < 1/4 \ \forall f \in D(e) \setminus \{e\}
  \tag{*}
  \label{eq:condition}
\end{equation}
and add this edge to~$E$ (if such an edge exists).  Note that this
changes the values of $D(f)$ and $w(f)$ for every edge $f$ on the path from $r$
to $v_e$.  We continue adding edges that satisfy \cref{eq:condition} to $E$
until every edge $e \in E(T_1) \setminus E$ satisfies $w(e) < 1/4$.  At the end
of the algorithm, we define $D_r$ to be the set of edges $f$ such that the path
from $r$ to $v_f$ contains no edge in $E$.  This set may be empty and is used
only in the analysis of the algorithm.

Next we prove that $|E| \le 4 \sum_{e \in E(T_1)} \tilde x_e$ and that $E$ is a
feasible solution of \cref{eq:ilp}.  As argued above, this implies that
$\F_E$ is a $4$-approximation of a MAF of $\T$.

\begin{lem}
  \label{lem:approximation-ratio}
  $\displaystyle|E| \le 4 \sum_{e \in E(T_1)} \tilde x_e$.
\end{lem}

\begin{proof}
  For every edge $e \in E$, let $D_1(e)$ and $D_2(e)$ be the values of $D(e)$ at
  the time when $e$ is added to $E$ and when the algorithm terminates,
  respectively.  Note that $D_2(e_1) \cap D_2(e_2) = \emptyset$, for any two
  distinct edges $e_1, e_2 \in E$.
  
  For every edge $e \in E$ and every edge $f \in D_2(e) \setminus \{e\}$, the
  path from $v_e$ to $v_f$ contains no edge in $E$ once the algorithm
  terminates.  Thus, this is also true at the time when the algorithm adds $e$
  to~$E$.  This shows that $f \in D_1(e)$, that is, $D_2(e) \subseteq D_1(e)$.

  If $D_1(e) \not\subseteq D_2(e)$, then there exists an edge $f \in D_1(e) \cap E$
  that is added to $E$ after $e$.  We have $w(f) \ge 1/4$ at the time we add $f$
  to $E$ and, since $f \in D_1(e)$, $w(f) < 1/4$ at the time we add $e$ to $E$.
  Adding edges to $E$ cannot increase $w(f)$ for any edge $f \in E(T_1)$.  Thus,
  $w(f)$ never increases.  This is a contradiction, and we must have $D_1(e)
  \subseteq D_2(e)$ and, therefore, $D_1(e) = D_2(e)$.
  
  Since $D_2(e_1) \cap D_2(e_2) = \emptyset$, for any two distinct edges $e_1,
  e_2 \in E$, we also have $D_1(e_1) \cap D_1(e_2) = \emptyset$, for any two
  such edges $e_1$ and $e_2$.  This implies that $\sum_{e \in E} w(e) = \sum_{e
  \in E} \sum_{f \in D_1(e)} \tilde x_f \le \sum_{e \in E(T_1)} \tilde x_e$.

  On the other hand, we have $w(e) \ge 1/4$, for every edge $e \in E$.  Thus,
  $\sum_{e \in E} w(e) \ge |E|/4$, that is, $|E| \le 4 \sum_{e \in E(T_1)}
  \tilde x_e$.
\end{proof}

\begin{lem}
  \label{lem:feasible-solution}
  $E$ is a feasible solution of \cref{eq:ilp}.
\end{lem}

\begin{proof}
  Similar to the proof of \cref{lem:approximation-ratio}, consider the values of
  $D(e)$ and $w(e)$ at the end of the algorithm, for every edge $e \in E(T_1)$.
  Then $w(e) < 1/4$, for every edge $e \notin E$.

  For any two leaves $a, b \in X$, let $u$ be their lowest common ancestor in
  $T_1$, that is, the vertex closest to $r$ that belongs to $T_1[\{a, b\}]$.  Then
  $T_1[\{a, b\}] = T_1[\{a, u\}] \cup T_1[\{u, b\}]$ and $\sum_{e \in T_1[\{a,
  b\}]} \tilde x_e = \sum_{e \in T_1[\{a, u\}]} \tilde x_e + \sum_{e \in
  T_1[\{u, b\}]} \tilde x_e$.  If $a = u$, then $T_1[\{a, u\}]$ contains no
  edges, so $\sum_{e \in T_1[\{a, u\}]} \tilde x_e = 0$.  If $a \ne u$ but
  $T_1[\{a, u\}] \cap E = \emptyset$, then $T_1[\{a, u\}] \subseteq D(e)$, where
  $e$ is the edge in $T_1[\{a, u\}]$ incident to $u$, so $\sum_{e \in T_1[\{a,
  u\}]} \tilde x_e \le w(e) < 1/4$.  This proves that $\sum_{e \in T_1[\{a,
  b\}]} \tilde x_e < 1/2$ for any two leaves $a, b \in X$ such that $T_1[\{a,
  b\}] \cap E = \emptyset$.
  
  For any quartet $Q = \{a, b, c, d\} \in \Q$, we have $\sum_{e \in \L(Q)}
  \tilde x_e \ge 1$ because $\tilde x$ is a feasible fractional solution of
  \cref{eq:ilp}.  Thus, if $T_1|_Q \ncong ab|cd$, then $\sum_{e \in
  T_1[\{a, b\}]} \tilde x_e \ge 1/2$ or $\sum_{e \in T_1[\{c, d\}]} \tilde x_e
  \ge 1/2$.  Therefore, $T_1[\{a, b\}] \cap E \ne \emptyset$ or $T_1[\{a, b\}]
  \cap E \ne \emptyset$, that is, $\L(Q) \cap E \ne \emptyset$.  Since this is
  true for every quartet $Q \in \Q$, $E$ is a feasible solution of
  \cref{eq:ilp}.
\end{proof}

Since the set of incompatible quartets $\Q$ in \cref{eq:ilp} contains at most
$\binom{n}{4} = O\bigl(n^4\bigr)$ quartets, the LP relaxation of \cref{eq:ilp}
has polynomial size.  Thus, it can be solved in polynomial time using the
ellipsoid algorithm \cite{ellipsoid} or any one of a number of more recent
interior point methods.  Together with
\cref{lem:approximation-ratio,lem:feasible-solution}, this shows the following
theorem.

\begin{thm}
  \label{thm:main}
  There exists a polynomial-time $4$-approximation algorithm for computing the MAF
  of a set of unrooted binary trees based on LP rounding.
\end{thm}

\section{A Family of Tight Inputs}

\label{sec:tight-inputs}

Next we prove that the integrality gap of \cref{eq:ilp} is $4 - o(1)$ even for
two input trees.  This implies that no approximation algorithm that uses an
optimal fractional solution of \cref{eq:ilp} (or any dual solution; see
\cref{sec:conclusions} for why this is important) as a lower bound on the
optimal solution can achieve an approximation ratio of $4 - \varepsilon$, for any
$\varepsilon > 0$.

\begin{lem}
  \label{lem:fractional-solution}
  There exists a fractional feasible solution $\tilde x$ of \cref{eq:ilp}
  that satisfies $\sum_{e \in E(T_1)} \tilde x_e = n/4$, where $n = |X|$.
\end{lem}

\begin{proof}
  For every leaf $v \in X$, let $e_v$ be the unique edge incident to $v$.  Now
  consider the fractional solution $\tilde x$ that sets $\tilde x_{e_v} = 1/4$
  for all $v \in X$, and $\tilde x_e = 0$ for any other edge.  Clearly, $\sum_{e
  \in E(T_1)} \tilde x_e = n/4$.

  To see that $\tilde x$ is feasible, observe that for every quartet $Q = \{a,
  b, c, d\} \in \Q$, we have $\{e_a, e_b, e_c, e_d\} \subseteq \L(Q)$.  Since
  $\tilde x_{e_v} = 1/4$ for every leaf $v \in X$, this shows that $\sum_{e \in
  \L(Q)} \tilde x_e \ge \tilde x_{e_a} + \tilde x_{e_b} + \tilde x_{e_c} +
  \tilde x_{e_d} = 1$.  Since this is true for every quartet $Q \in \Q$, $\tilde
  x$ is feasible.
\end{proof}

\begin{lem}
  \label{lem:integral-solution}
  There exists an infinite family of pairs of trees such that any agreement
  forest of any pair $(T_1, T_2)$ in this family has $n - o(n)$ components,
  where $n$ is the size of the label set $X$ of $T_1$ and $T_2$.
\end{lem}

\begin{proof}
  Let $\ell \ge 4$.  Consider two trees $T_1$ and $T_2$ with leaf set $X =
  \{(i,j) \mid 1 \le i, j \le \ell\}$.  Thus, $n = |X| = \ell^2$, that is, $\ell
  = \sqrt{n}$.  Both $T_1$ and $T_2$ are caterpillars, that is, the internal
  vertices of each tree form a path, which we call the \emph{spine} of the
  caterpillar.  In $T_1$, the leaves are attached to this path sorted by their
  $i$-components and then by their $j$-components.  In~$T_2$, they are sorted by
  their $j$-components and then by their $i$-components.  For an example, see
  \cref{fig:trees}.

  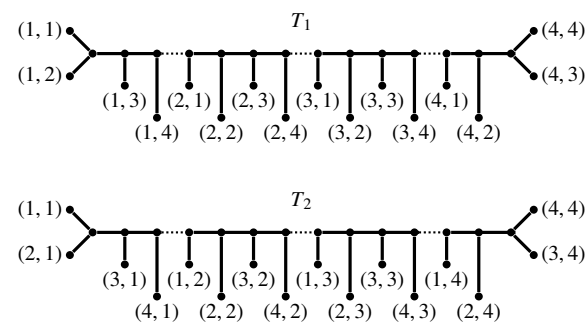
\begin{figure}[b]
    \centering
    \begin{tikzpicture}[
      x=4.25mm,y=4.25mm,
      vertex/.style={inner sep=0pt,circle,fill=black,minimum size=3pt},
      edge/.style={draw,very thick},
      sep/.style={draw,densely dotted,thick},
    ]
      \footnotesize
      \path node [vertex] (2) {}
      foreach \i in {3,...,15} {
        ++(0:1) node [vertex] (\i) {}
      }
      foreach \i in {3,5,...,14} {
        (\i) +(270:1) node [vertex] (x\i) {}
      }
      foreach \i in {4,6,...,14} {
        (\i) +(270:2) node [vertex] (x\i) {}
      }
      (2) +(135:1) node [vertex] (x1) {}
      +(225:1) node [vertex] (x2) {}
      (15) +(45:1) node [vertex] (x16) {}
      +(315:1) node [vertex] (x15) {};
      \path [edge]
      foreach \i [evaluate=\i as \j using \i+1] in {2,3,5,6,7,9,10,11,13,14} {
        (\i) -- (\j)
      }
      foreach \i in {3,...,14} {
        (\i) -- (x\i)
      }
      (x1) -- (2) -- (x2)
      (x15) -- (15) -- (x16);
      \path [sep]
      foreach \i [evaluate=\i as \j using \i+1] in {4,8,12} {
        (\i) -- (\j)
      };
      \path
      (x1) node [anchor=east] {$(1,1)$}
      (x2) node [anchor=east] {$(1,2)$}
      (x3) node [anchor=north] {$(1,3)$}
      (x4) node [anchor=north] {$(1,4)$}
      (x5) node [anchor=north] {$(2,1)$}
      (x6) node [anchor=north] {$(2,2)$}
      (x7) node [anchor=north] {$(2,3)$}
      (x8) node [anchor=north] {$(2,4)$}
      (x9) node [anchor=north] {$(3,1)$}
      (x10) node [anchor=north] {$(3,2)$}
      (x11) node [anchor=north] {$(3,3)$}
      (x12) node [anchor=north] {$(3,4)$}
      (x13) node [anchor=north] {$(4,1)$}
      (x14) node [anchor=north] {$(4,2)$}
      (x15) node [anchor=west] {$(4,3)$}
      (x16) node [anchor=west] {$(4,4)$};
      \node [anchor=south,yshift=6pt] at (barycentric cs:8=0.5,9=0.5) {$T_1$};
    \end{tikzpicture}\\[\bigskipamount]
    \begin{tikzpicture}[
      x=4.25mm,y=4.25mm,
      vertex/.style={inner sep=0pt,circle,fill=black,minimum size=3pt},
      edge/.style={draw,very thick},
      sep/.style={draw,densely dotted,thick},
    ]
      \footnotesize
      \path node [vertex] (2) {}
      foreach \i in {3,...,15} {
        ++(0:1) node [vertex] (\i) {}
      }
      foreach \i in {3,5,...,14} {
        (\i) +(270:1) node [vertex] (x\i) {}
      }
      foreach \i in {4,6,...,14} {
        (\i) +(270:2) node [vertex] (x\i) {}
      }
      (2) +(135:1) node [vertex] (x1) {}
      +(225:1) node [vertex] (x2) {}
      (15) +(45:1) node [vertex] (x16) {}
      +(315:1) node [vertex] (x15) {};
      \path [edge]
      foreach \i [evaluate=\i as \j using \i+1] in {2,3,5,6,7,9,10,11,13,14} {
        (\i) -- (\j)
      }
      foreach \i in {3,...,14} {
        (\i) -- (x\i)
      }
      (x1) -- (2) -- (x2)
      (x15) -- (15) -- (x16);
      \path [sep]
      foreach \i [evaluate=\i as \j using \i+1] in {4,8,12} {
        (\i) -- (\j)
      };
      \path
      (x1) node [anchor=east] {$(1,1)$}
      (x2) node [anchor=east] {$(2,1)$}
      (x3) node [anchor=north] {$(3,1)$}
      (x4) node [anchor=north] {$(4,1)$}
      (x5) node [anchor=north] {$(1,2)$}
      (x6) node [anchor=north] {$(2,2)$}
      (x7) node [anchor=north] {$(3,2)$}
      (x8) node [anchor=north] {$(4,2)$}
      (x9) node [anchor=north] {$(1,3)$}
      (x10) node [anchor=north] {$(2,3)$}
      (x11) node [anchor=north] {$(3,3)$}
      (x12) node [anchor=north] {$(4,3)$}
      (x13) node [anchor=north] {$(1,4)$}
      (x14) node [anchor=north] {$(2,4)$}
      (x15) node [anchor=west] {$(3,4)$}
      (x16) node [anchor=west] {$(4,4)$};
      \node [anchor=south,yshift=6pt] at (barycentric cs:8=0.5,9=0.5) {$T_2$};
    \end{tikzpicture}
    \caption{The two trees $T_1$ and $T_2$ in the proof of
    \cref{lem:integral-solution} for $\ell = 4$.}
    \label{fig:trees}
  \end{figure}

  Now we call an edge $e \in E(T_1)$ \emph{separating} if there exists an index
  $1 \le i < \ell$ such that $e$ belongs to the path from $(i,j)$ to $(i+1,j')$,
  for all $1 \le j, j' \le \ell$.  Similarly, we call an edge $e \in E(T_2)$
  separating if there exists an index $1 \le j < \ell$ such that $e$ belongs to
  the path from $(i,j)$ to $(i',j+1)$, for all $1 \le i, i' \le \ell$.  The
  separating edges are shown dotted in \cref{fig:trees}.  Note that there are
  $\ell - 1$ separating edges in $T_1$, and $\ell - 1$ separating edges in
  $T_2$, $2\ell - 2$ separating edges in total.

  Now assume that $\F = \{Y_1, \ldots, Y_k\}$ is an AF of $T_1$ and $T_2$.  For
  each component $Y_h$, let $S_{Y_h}^1$ be the set of separating edges of $T_1$
  contained in $T_1[Y_h]$, let $S_{Y_h}^2$ be the set of separating edges of
  $T_2$ contained in $T_2[Y_h]$, and let $S_{Y_h} = S_{Y_h}^1 \cup S_{Y_h}^2$.
  Since no two components of $\F$ overlap in either $T_1$ or $T_2$, we have
  $S_{Y_h} \cap S_{Y_{h'}} = \emptyset$, for any two distinct components $Y_h,
  Y_{h'} \in \F$, that is, $\sum_{h=1}^k |S_{Y_h}| \le 2\ell - 2$.  Next we
  prove that $|Y_h| \le |S_{Y_h}| + 1$, for all $Y_h \in \F$.  This implies that
  $n = \sum_{h=1}^k |Y_h| \le \sum_{h=1}^k (|S_{Y_h}| + 1) = \sum_{h=1}^k
  |S_{Y_h}| + k \le 2\ell - 2 + k$, so $k \ge n - 2\ell + 2 = n - 2\sqrt{n} + 2
  = n - o(n)$, and the lemma follows.

  Consider a component $Y_h = \{x_1, \ldots, x_s\}$ of $\F$; let $x_r = (i_r,
  j_r)$, for all $1 \le r \le s$; and assume that the leaves in $Y_h$ are
  indexed in the order in which they occur along $T_1$.  For each index
  $1 \le r < s$, we choose a separating edge $e_r$ of $T_1$ or $T_2$ such that
  $e_r \in S_{Y_h}$ and, for all $1 \le r_1 < r_2 < s$, $e_{r_1} \ne e_{r_2}$.
  This immediately implies that $|Y_h| = s \le |S_{y_h}| + 1$.
  
  If $i_r < i_{r+1}$, then we choose $e_r$ to be the $i_r$th separating edge in
  $T_1$.  Otherwise, we must have $i_r = i_{r+1}$ and $j_r < j_{r+1}$, and we
  choose $e_r$ to be the $j_r$th separating edge in $T_2$.  Since the leaves in
  $Y_h$ are indexed in the order in which they occur along $T_1$, we have
  $i_{r_1} \ne i_{r_2}$ for any two indices $r_1 \ne r_2$ such that $i_{r_1} <
  i_{r_1 + 1}$ and $i_{r_2} < i_{r_2 + 1}$.  Thus, the separating edges chosen
  from $T_1$ are all distinct.

  Next assume that there exist two indices $r_1 < r_2$ such that $e_{r_1}$ and
  $e_{r_2}$ are the same separating edge from $T_2$.  Then $i_{r_1} = i_{r_1 +
  1}$, $j_{r_1} < j_{r_1+1}$, $i_{r_2} = i_{r_2 + 1}$, $j_{r_2} < j_{r_2 + 1}$,
  and $j_{r_1} = j_{r_2}$.  Since $j_{r_1} = j_{r_2}$, we must have $i_{r_1} \ne
  i_{r_2}$ and, therefore, $i_{r_1+1} \ne i_{r_2+1}$.  Thus, $Q = \{x_{r_1},
  x_{r_1+1}, x_{r_2}, x_{r_2+1}\} \subseteq Y_h$ is a quartet that satisfies
  $T_1|_Q \cong x_{r_1}x_{r_1+1}|x_{r_2}x_{r_2+1}$ and $T_2|_Q \cong
  x_{r_1}x_{r_2}|x_{r_1+1}x_{r_2+1}$, a contradiction because $\F$ is an AF of
  $T_1$ and $T_2$ and $Y_h$ is a component of $\F$.  Thus, all separating edges
  chosen from $T_2$ are also distinct.  This finishes the proof.
\end{proof}

\cref{lem:fractional-solution,lem:integral-solution} together prove the
following theorem.

\begin{thm}
  \label{thm:integrality-gap}
  The integrality gap of \cref{eq:ilp} is at least $4 - o(1)$ even if $|\T| =
  2$.
\end{thm}

\begin{proof}
  Consider any pair in the family of tree pairs provided by
  \cref{lem:integral-solution}.  The optimal integral solution of \cref{eq:ilp}
  for this pair of trees has objective function value $n - o(n)$.  By
  \cref{lem:fractional-solution}, the optimal fractional solution has objective
  function value at most $n/4$.  Thus, the integrality gap is at least $\frac{n
  - o(n)}{n/4} = 4 - o(1)$.
\end{proof}

\section{Conclusions}

\label{sec:conclusions}

\cref{thm:main,thm:integrality-gap} are significant for at least three reasons.

First, \cref{thm:main} matches the approximation ratio of the significantly
more complex 4-approximation algorithm of \cite{chen}, which is based on
purely combinatorial arguments though.

Second, in \cite{dmp}, \cref{eq:ilp} and the ILP version of its dual were used
to prove that there exists a kernel of size $O(k \lg k)$ for the 2-state maximum
parsimony distance of two trees.  This dual is

\begin{equation}
  \begin{gathered}
    \textrm{Maximize}\ \sum_{Q \in \Q} y_Q\\
    \begin{aligned}
      \textrm{s.t.}\ \sum_{e \in \L(Q)} y_Q &\le 1 && \forall e \in E(T_1)\\
      y_Q &\in \{0, 1\} && \forall Q \in \Q.
    \end{aligned}
  \end{gathered}
  \label{eq:dual}
\end{equation}

The key to bounding the size of the kernel was to prove that the gap between
optimal integral solutions of \cref{eq:ilp,eq:dual} is at most $O(\lg k)$.  If
both \cref{eq:ilp,eq:dual} have a constant integrality gap, then this proves
that the kernel in \cite{dmp} is in fact a linear kernel for the 2-state maximum
parsimony distance.  \cref{thm:main} proves the first half of this conjecture.

Finally, our hope was to use primal-dual arguments based on
\cref{eq:ilp,eq:dual} to obtain a 2-approximation algorithm, or at least a
$c$-approximation algorithm with $c < 3$, for the TBR distance of two unrooted
binary trees.  \cref{thm:integrality-gap} proves that this is impossible.  Thus,
if there exists a 2-approximation algorithm for the TBR distance of two unrooted
binary trees, it needs to be based on a different ILP formulation, possibly one
mimicking the ILP in \cite{red-blue}, which was used to prove that the algorithm
in \cite{red-blue} outputs a 2-approximation of a MAF for two \emph{rooted}
trees.

\bibliographystyle{abbrv}
\bibliography{ref}

\begin{thebibliography}{10}

\bibitem{allensteel}
B.~L. Allen and M.~Steel.
\newblock Subtree transfer operations and their induced metrics on evolutionary
  trees.
\newblock {\em Annals of Combinatorics}, 5(1):1--15, Jun 2001.

\bibitem{chen2013}
J.~Chen, J.-H. Fan, and S.-H. Sze.
\newblock Parameterized and approximation algorithms for the {MAF} problem in
  multifurcating trees.
\newblock In {\em Proceedings of the 39th International Workshop on
  Graph-Theoretic Concepts in Computer Science (WG 2013)}, pages 152--164.
  Springer, 2013.

\bibitem{chen}
J.~Chen, F.~Shi, and J.~Wang.
\newblock Approximating maximum agreement forest on multiple binary trees.
\newblock {\em Algorithmica}, 76:867--889, 2016.

\bibitem{colonius77}
H.~Colonius and H.-H. Schulze.
\newblock {\em Trees constructed from empirical relations}.
\newblock Inst. f{\"u}r Psychologie d. Techn. Univ., 1977.

\bibitem{colonius81}
H.~Colonius and H.-H. Schulze.
\newblock Tree structures for proximity data.
\newblock {\em British Journal of Mathematical and Statistical Psychology},
  34(2):167--180, 1981.

\bibitem{dmp}
E.~Deen, L.~van Iersel, R.~Janssen, M.~Jones, Y.~Murakami, and N.~Zeh.
\newblock A near-linear kernel for bounded-state parsimony distance.
\newblock {\em Journal of Computer and System Sciences}, 140:103477, 2024.

\bibitem{hallett}
M.~Hallett and C.~McCartin.
\newblock A faster fpt algorithm for the maximum agreement forest problem.
\newblock {\em Theory of Computing Systems}, 41(3):539--550, 2007.

\bibitem{hein}
J.~Hein, T.~Jiang, L.~Wang, and K.~Zhang.
\newblock On the complexity of comparing evolutionary trees.
\newblock {\em Discrete Applied Mathematics}, 71(1-3):153--169, 1996.

\bibitem{kelk2022}
S.~Kelk, S.~Linz, and R.~Meuwese.
\newblock Deep kernelization for the tree bisection and reconnection (tbr)
  distance in phylogenetics.
\newblock {\em Journal of Computer and System Sciences}, 142:103519, 2024.

\bibitem{ellipsoid}
L.~G. Khachiyan.
\newblock A polynomial algorithm for linear programming.
\newblock {\em Doklady Akademiia Nauk USSR}, 244:1093--1096, 1979.
\newblock in Russian.

\bibitem{linz-cluster}
S.~Linz and C.~Semple.
\newblock A cluster reduction for computing the subtree distance between
  phylogenies.
\newblock {\em Annals of Combinatorics}, 15(3):465--484, Sept. 2011.

\bibitem{red-blue}
N.~Olver, F.~Schalekamp, S.~van~der Ster, L.~Stougie, and A.~van Zuylen.
\newblock A duality based 2-approximation algorithm for maximum agreement
  forest.
\newblock {\em Mathematical Programming}, 198(1):811--853, 2023.

\bibitem{van}
R.~van Wersch, S.~Kelk, S.~Linz, and G.~Stamoulis.
\newblock Reflections on kernelizing and computing unrooted agreement forests.
\newblock {\em Annals of Operations Research}, 309(1):425--451, 2022.

\bibitem{whidden}
C.~Whidden and N.~Zeh.
\newblock A unifying view on approximation and {FPT} of agreement forests.
\newblock In {\em Proceedings of the 9th International Workshop on Algorithms
  in Bioinformatics (WABI 2009)}, volume 5724 of {\em Lecture Notes in Computer
  Science}, pages 390--402. Springer, 2009.

\end{thebibliography}

\end{document}